\documentclass[1p,final]{elsarticle}
\usepackage{amsfonts,color,morefloats,pslatex}
\usepackage{amssymb,amsthm, amsmath,latexsym}

\allowdisplaybreaks[4]

\newtheorem{thm}{Theorem}[section]
\newtheorem{cor}[thm]{Corollary}
\newtheorem{lem}[thm]{Lemma}

\newtheorem{rem}[thm]{Remark}

\newtheorem{defn}[thm]{Definition}

\numberwithin{equation}{section}
\newcommand{\Tr}{\operatorname{Tr}}

\begin{document}

\begin{frontmatter}



\title{Nearly optimal codebooks based on generalized Jacobi sums}

\author{Ziling Heng}
\ead{zilingheng@163.com, hengziling@ust.hk}
\address{Department of Computer Science and Engineering,
The Hong Kong University of Science and Technology,
Clear Water Bay, Kowloon, Hong Kong, China}

\begin{abstract}
Codebooks with small inner-product correlation are applied in many practical applications including direct spread code division multiple access (CDMA) communications, space-time codes and compressed sensing. It is extremely difficult to construct codebooks achieving the Welch bound or the Levenshtein bound. Constructing nearly optimal codebooks such that the ratio of its maximum cross-correlation amplitude to the corresponding bound approaches 1 is also an interesting research topic.  In this paper, we firstly study a family of interesting character sums called generalized Jacobi sums over finite fields. Then we apply the generalized Jacobi sums and their related character sums to obtain two infinite classes of  nearly optimal codebooks with respect to the Welch or Levenshtein bound. The codebooks can be viewed as generalizations of some known ones and contain new ones with very flexible parameters.
\end{abstract}


\end{frontmatter}

\section{Introduction}
 Codebooks(also called signal sets) with small inner-product correlation are usually used to distinguish among the signals of different users in CDMA systems. An $(N,K)$ codebook $\mathcal{C}$ is a set $\{\textbf{c}_{0},\textbf{c}_{1},...,\textbf{c}_{N-1}\}$, where each codeword $\textbf{c}_{l},0\leq l \leq N-1$, is a unit norm $1\times K$ complex vector over an alphabet $A$. The alphabet size is the number of elements in $A$. The maximum cross-correlation amplitude, which is a very important performance measure of a codebook in practical applications, of an $(N,K)$ codebook $\mathcal{C}$ is defined by
\begin{eqnarray*}
I_{\max}(\mathcal{C})= \max\limits_{0 \leq i<j \leq N-1} \mid \textbf{c}_{i}\textbf{c}_{j}^{H} \mid,
\end{eqnarray*}
where $\textbf{c}^{H}$ denotes the conjugate transpose of a complex vector $\textbf{c}$.  Minimizing the  maximal cross-correlation
amplitude of a codebook is an important problem as it can approximately optimize many performance
metrics such as outage probability, average signal-to-noise ratio and symbol error probability for multiple-antenna transmit beamforming from limited-rate feedback \cite{LHS, MSEA}. Besides, minimizing $I_{\max}(\mathcal{C})$  is equivalent to minimizing
the block error probability in the context of unitary space-time modulations \cite{HMRSU}. If $N\geq K$, codebooks are also called \emph{frames} and a codebook
$\mathcal{C}$ with minimal $I_{\max}(\mathcal{C})$ is referred to as a \emph{Grassmannian frame}.
Furthermore, minimizing $I_{\max}(\mathcal{C})$ of finite frames brings to minimal reconstruction error in multiple description coding over erasure channels \cite{SH}.

 For a given $K$, it is desirable to construct an $(N,K)$ codebook with $N$ being as large as possible and $I_{\max}(\mathcal{C})$ being as small as possible simultaneously. There exist some bounds among the parameters $N$, $K$ and $I_{\max}(\mathcal{C})$ of a codebook $\mathcal{C}$.

 The Welch bound is given as follows.
\begin{lem}[Welch bound]\label{lem-I.1} \cite{W}
For any $(N,K)$ codebook $\mathcal{C}$ with $N\geq K$,
\begin{eqnarray}\label{eqn-I.1}
I_{\max}(\mathcal{C})\geq \sqrt{\frac{N-K}{(N-1)K}}.
\end{eqnarray}
In addition, the equality in (\ref{lem-I.1}) is achieved if and only if
\begin{eqnarray*}|\textbf{c}_{i}\textbf{c}_{j}^{H}|=\sqrt{\frac{N-K}{(N-1)K}}
\end{eqnarray*}
for all pairs $(i,j)$ with $i\neq j$.
\end{lem}

If a codebook $\mathcal{C}$ achieves the Welch bound in (\ref{lem-I.1}), which is denoted by $I_{W}$, we call it a maximum-Welch-bound-equality (MWBE) codebook \cite{XZG}. An MWBE codebook is referred to as an
equiangular tight frame \cite{TKK}. The reader is referred to \cite{DHC, STDH} for the connection of MWBE codebooks and equiangular tight frames.
MWBE codebooks are also applied in many practical  applications including CDMA communications, space-time codes and compressed sensing \cite{MM, TK, TKK}. To our knowledge,
only a few constructions of MWBE codebooks were reported in literature. We list them in the following.
\begin{enumerate}
\item[(1)] In \cite{S,XZG}, optimal $(N,N)$ and $(N,N-1)$ codebooks with $N>1$ were generated based on the (inverse) discrete Fourier transform matrix or ideal two-level autocorrelation sequences. Specific constructions of optimal $(N,N-1)$ codebooks were also given in \cite{GK, DO}. In fact, optimal $(N,N)$ codebooks are the same as orthonormal bases.
\item[(2)] In \cite{CHS,SH}, optimal $(N,K)$ codebooks from conference matrices were given when $N=2K=2^{d+1}$ with $d$ being a positive integer and $N=2K=p^{d}+1$ with $p$ being a prime number and $d$ being a positive integer.
\item[(3)] In \cite{D,DF2,XZG}, optimal $(N,K)$ codebooks were constructed with cyclic difference sets in the Abelian group $(\mathbb{Z}_{N},+)$ or the additive group of finite fields or Abelian groups in general.
\item[(4)] In \cite{FMT}, the authors constructed optimal $(N,K)$ codebooks from $(2,k,v)$-Steiner systems.
\item[(5)] In \cite{FMJ162, Rahimi}, graph theory and finite geometries were applied to
study MWBE codebooks.
\end{enumerate}

According to \cite{SH}, the Welch bound on $I_{\max}(\mathcal{C})$ of a codebook  $\mathcal{C}$ is not tight when $N>K(K+1)/2$ for real codebooks and $N>K^{2}$ for all codebooks. The following Levenshtein bound turns  out to be tighter than the Welch bound in these cases.
\begin{lem}[Levenshtein bound]\label{lem-I.2} \cite{L}
For any real-valued codebook $\mathcal{C}$ with $N>K(K+1)/2$,
\begin{eqnarray}\label{eqn-I.2}
I_{\max}(\mathcal{C})\geq \sqrt{\frac{3N-K^{2}-2K}{(N-K)(K+2)}}.
\end{eqnarray}
For any complex-valued codebook $\mathcal{C}$ with $N>K^{2}$,
\begin{eqnarray}\label{eqn-I.3}
I_{\max}(\mathcal{C})\geq \sqrt{\frac{2N-K^{2}-K}{(N-K)(K+1)}}.
\end{eqnarray}
\end{lem}

In general, it is very hard to construct codebooks achieving the Levenstein bound, which is denoted by $I_{L}$
(the right-hand side of (\ref{eqn-I.2}) or (\ref{eqn-I.3})). There are only a few known optimal constructions of codebooks achieving the Levenshtein bound. These optimal
codebooks were constructed from Kerdock codes \cite{CCKS,XDM}, perfect nonlinear functions \cite{DY}, bent functions over finite fields \cite{ZDL}, and bent functions over Galois rings \cite{HY}. Codebooks achieving the Levenshtein bound are used in quantum physics and the design of spreading sequences for CDMA and sets of mutually unbiased bases \cite{DY,WF}.

 Since it is very difficult to construct optimal codebooks, there have been a lot of attempts to construct a codebook nearly meeting the Welch bound or the Levenshtein bound with equality, i.e., $I_{\max}(\mathcal{C})$ is slightly higher than the bound equality, but asymptotically achieves it for large enough $N$. We follow the following definition throughout this paper.
 \begin{defn}\label{defn-intro} An $(N,K)$ codebook $\mathcal{C}$ is said to be nearly optimal if one of the following holds:
 \begin{enumerate}
 \item[$\bullet$] $\lim_{N \rightarrow \infty}\frac{I_{\max}(\mathcal{C})}{I_{W}}=1$ for any $(N,K)$ codebook $\mathcal{C}$ with $N\geq K$;
 \item[$\bullet$] $\lim_{N \rightarrow \infty}\frac{I_{\max}(\mathcal{C})}{I_{L}}=1$ for any real-valued codebook $\mathcal{C}$ with $N>K(K+1)/2$ or any complex-valued codebook $\mathcal{C}$ with $N>K^{2}$.
 \end{enumerate}\end{defn}
We remark that Definition \ref{defn-intro} has been actually used in \cite{HDY, HW, TZZ, XDM, Y, ZT}, though it was not  explicitly given before. We summarize
some well known nearly optimal codebooks in the following.
\begin{enumerate}
\item[$\diamondsuit$] In \cite{HW}, new constructions of $(N,K)$ codebooks nearly meeting the Welch bound were proposed based on difference sets and the product of Abelian groups.
\item[$\diamondsuit$] In \cite{ZT}, a construction of $(uv+k,k)$ codebook $\mathcal{C}$ with $I_{\max}(\mathcal{C})=\sqrt{\frac{1}{k}}$ was given from $(v,u,k,\lambda)$ relative difference set in an abelian group $G$ relative to a subgroup $H$ of $G$. Some specific nearly optimal codebooks were obtained by this construction.
\item[$\diamondsuit$] In \cite{TZZ}, the authors used Gauss sums to construct nearly optimal $(q^2-1,q-1)$ codebook $\mathcal{C}$ with $I_{\max}(\mathcal{C})=\frac{\sqrt{q}}{q-1}$, where $q$ is a power of a prime.
\item[$\diamondsuit$] In \cite{Y}, the authors constructed a new nearly optimal $(N,K)$ partial Fourier codebook $\mathcal{C}$ with $I_{\max}(\mathcal{C})=\frac{1}{\sqrt{K}}$, where $N=K^2-1$ and $K=p^k$ for any prime $p$ and a positive integer $k$.
\item[$\diamondsuit$] In \cite[Theorem 7]{XDM}, nearly optimal $(2^{2m}+2^m,2^m)$ codebooks $\mathcal{C}$ with $I_{\max}(\mathcal{C})=\sqrt{\frac{1}{2^{m-1}}}$ with respect to the Levenshtein bound were presented based on binary codes.
\item[$\diamondsuit$] In \cite{HDY}, new codebooks $\mathcal{C}$ with parameters $((q-1)^{k}+n,n)$ with $I_{\max}(\mathcal{C})=\frac{q^{\frac{k+1}{2}}}{(q-1)^{k}+(-1)^{k+1}}$, and new codebooks $\mathcal{C}'$ with parameters $((q-1)^{k}+q^{k-1},q^{^{k-1}})$ with $I_{\max}(\mathcal{C}')=\frac{q^{\frac{k+1}{2}}}{(q-1)^{k}+(-1)^{k+1}}$ were constructed with multiplicative characters over finite fields, where $n=\frac{1}{q}((q-1)^{k}+(-1)^{k+1})$ and $q$ is a power of a prime.
\end{enumerate}
Besides, there are also some constructions of codebooks with relatively small maximum
cross-correlation amplitude in \cite{DF, LYH, S, ZF1}.

The purpose of this paper is to construct nearly optimal codebooks based on some interesting character sums. We firstly introduce a family of character sums called generalized Jacobi sums which can be viewed as a generalization of the classical Jacobi sums. Based on the generalized Jacobi sums and their related character sums, two classes of nearly optimal codebooks with very flexible parameters are constructed. Our constructions produce codebooks with new parameters compared with known ones in literature. In particular, our main results contain those in \cite{HDY} as special cases.
\section{Mathematical Foundations}

In this section, we recall some necessary mathematical foundations on characters, Jacobi sums and Gauss sums
over finite fields. They will play important roles in our constructions of codebooks.

In this paper, we always assume that $p$ is a prime number and $q=p^{m}$ with $m$ being a positive integer. Let $\mathbb{F}_{q}$ denote the finite field with $q$ elements. Let $\alpha$ be a primitive element of $\mathbb{F}_{q}$. Let $\Tr_{q/p}$ be the trace function from $\mathbb{F}_{q}$ to $\mathbb{F}_{p}$ defined by
$$\Tr_{q/p}(x)=\sum_{j=0}^{m-1} x^{p^j}.$$

\subsection{Characters over finite fields}

In this section, we recall both additive and multiplicative characters over finite fields.

\begin{defn}
An additive character of $\mathbb{F}_{q}$ is a mapping $\chi$ from $\mathbb{F}_{q}$ to the set $\mathbb{C}^{*}$ of nonzero complex numbers such that $\chi(x+y)=\chi(x)\chi(y)$ for any $(x,y)\in \mathbb{F}_{q}\times \mathbb{F}_{q}$.
\end{defn}

It is well known that every additive character of $\mathbb{F}_{q}$ can be expressed as
$$\chi_{a}(x)=\zeta_{p}^{\Tr_{q/p}(ax)},\ x\in \mathbb{F}_{q},$$
where $\zeta_{p}$ is a primitive $p$-th root of complex unity. In particular, we call $\chi_0$ the trivial additive character and $\chi_1$ the canonical additive character of $\mathbb{F}_{q}$. The orthogonal relation of additive characters (see \cite{LN}) is given by
$$ \sum_{x\in \mathbb{F}_{q}}\chi_1(ax)=\left\{
\begin{array}{ll}
  q,   &      \mbox{if}\ a=0,\\
0, & \mbox{otherwise}.
\end{array} \right. $$

\begin{defn}\label{def-II}
A multiplicative character of $\mathbb{F}_{q}$ is a nonzero function $\psi$ from $\mathbb{F}_{q}^{*}$ to the set  $\mathbb{C}^{*}$ of nonzero complex numbers such that $\psi(xy)=\psi(x)\psi(y)$ for any $x,y\in \mathbb{F}_{q}^{*}$, where $\mathbb{F}_{q}^{*}=\mathbb{F}_{q}\backslash\{0\}$.
\end{defn}

The multiplicative characters of $\mathbb{F}_{q}$ can be expressed as follows \cite{LN}. For $j=0,1,\cdots,q-2$, the functions $\psi_{j}$ defined by
 $$\psi_{j}(\alpha^{k})=\zeta_{q-1}^{jk},\mbox{ for }k=0,1,\cdots,q-2,$$
are all the multiplicative characters of $\mathbb{F}_{q}$, where $\zeta_{h}=e^{\frac{2\pi \sqrt{-1}}{h}}$ denotes the $h$-th root of complex unity. If $j=0$, we have $\psi_{0}(x)=1$ for any $x\in\mathbb{F}_{q}^{*}$ and $\psi_{0}$  is called the trivial multiplicative character of $\mathbb{F}_{q}$.

For two multiplicative characters $\psi,\psi'$, we define their multiplication by setting $\psi\psi'(x)=\psi(x)\psi'(x)$ for all $x\in \mathbb{F}_{q}^{*}$. Let $\widehat{\mathbb{F}}_{q}^{*}$ be the set of all multiplicative characters of $\mathbb{F}_{q}$. Let $\overline{\psi}$ denote the conjugate character of $\psi$ by setting $\overline{\psi}(x)=\overline{\psi(x)}$, where $\overline{\psi(x)}$ denotes the complex conjugate of $\psi(x)$. It is easy to verify that $\psi^{-1}=\overline{\psi}$. Then $\widehat{\mathbb{F}}_{q}^{*}$ forms a group under the  multiplication of characters. Furthermore, $\widehat{\mathbb{F}}_{q}^{*}$ is isomorphic to $\mathbb{F}_{q}^{*}$.

For a multiplicative character $\psi$ of $\mathbb{F}_q$, the orthogonal relation (see \cite{LN}) of it is given by
 $$ \sum_{x\in \mathbb{F}_{q}^{*}}\psi(x)=\left\{
\begin{array}{ll}
  q-1,   &      \mbox{if}\ \psi=\psi_0,\\
0, & \mbox{otherwise}.
\end{array} \right. $$

\subsection{Jacobi sums}

 We now extend the definition of a multiplicative character $\psi$ by setting
 \begin{eqnarray}\label{eqn-extend}\psi(0)=\left\{
\begin{array}{ll}
1, & \mbox{if }\psi=\psi_{0},\\
0, & \mbox{if }\psi\neq \psi_0.
\end{array}
\right.\end{eqnarray}
Then the property that $\psi(xy)=\psi(x)\psi(y)$ holds for all $x,y\in \mathbb{F}_{q}$. With this definition, we deduce that
\begin{eqnarray}\label{eqn-sum} \sum_{x\in \mathbb{F}_{q}}\psi(x)=\left\{
\begin{array}{ll}
  q,   &      \mbox{if}\ \psi=\psi_0,\\
0, & \mbox{otherwise}.
\end{array} \right. \end{eqnarray}

\begin{defn}\cite[p. 205]{LN}
Let $\lambda_{1},\ldots,\lambda_{k}$ be $k$ multiplicative characters of $\mathbb{F}_{q}$.
The sum
$$
J(\lambda_{1},\ldots,\lambda_{k})=\sum_{c_{1}+\cdots+c_{k}=1 \atop c_1, \cdots, c_k \in \mathbb{F}_{q}}\lambda_{1}(c_1)\cdots \lambda_k(c_k),
$$
is called a \textit{Jacobi sum} in $\mathbb{F}_{q}$.
\end{defn}

Jacobi sums are very useful in coding theory, sequence design and cryptography. For any $a\in \mathbb{F}_{q}^{*}$, more generally, we define the sum
\begin{eqnarray*}
J_{a}(\lambda_{1},\ldots,\lambda_{k})=\sum_{c_{1}+\cdots+c_{k}=a}\lambda_{1}(c_1)\cdots \lambda_k(c_k), \end{eqnarray*}
where the summation extends over all $k$-tuples $(c_1,\ldots,c_k)$ of elements of $\mathbb{F}_{q}$ satisfying $c_{1}+\cdots+c_{k}=a$. Hence, $J_{1}(\lambda_{1},\ldots,\lambda_{k})=J(\lambda_{1},\ldots,\lambda_{k})$. It was shown in \cite[p. 205]{LN} that
\begin{eqnarray*}
J_{a}(\lambda_{1},\ldots,\lambda_{k})=(\lambda_{1}\cdot\cdot\cdot\lambda_{k})(a)J(\lambda_{1},\ldots,\lambda_{k}).
\end{eqnarray*}
Therefore, $|J_{a}(\lambda_{1},\ldots,\lambda_{k})|=|J(\lambda_{1},\ldots,\lambda_{k})|$. In \cite{LN}, the values of $|J(\lambda_{1},\ldots,\lambda_{k})|$ were determined for several cases.
\subsection{Gauss sums}
Let $\psi$ be a multiplicative character and $\chi$ an additive character of $\Bbb F_{q}$. The \emph{Gauss sum} $G(\psi,\chi)$ is defined by
\begin{eqnarray*}G(\psi,\chi)=\sum_{x\in \Bbb F_{q}^{*}}\psi(x)\chi(x).\end{eqnarray*}

The explicit value of $G(\psi,\chi)$ is very difficult to determine in general. However, its absolute value is known as follows.
\begin{lem}\label{lem-gausssum}\cite[Th. 5.11]{LN}
Let $\psi$ be a multiplicative character and $\chi$ an additive character of $\Bbb F_{q}$. Then $G(\psi,\chi)$ satisfies
\begin{eqnarray*} G(\psi,\chi)=\left\{
\begin{array}{ll}
  q-1,   &      \mbox{if}\ \psi=\psi_0,\chi=\chi_0,\\
-1, & \mbox{if}\ \psi=\psi_0,\chi\neq\chi_0,\\
0, &  \mbox{if}\ \psi\neq\psi_0,\chi=\chi_0.\\
\end{array} \right. \end{eqnarray*}
If $\psi\neq\psi_0,\chi\neq\chi_0$, then
\begin{eqnarray*}|G(\psi,\chi)|=\sqrt{q}.\end{eqnarray*}
\end{lem}

If we consider the extended definition of $\psi$ in Equation (\ref{eqn-extend}), then the extended Gauss sum can be defined as
\begin{eqnarray*}\widehat{G}(\psi,\chi)=\sum_{x\in \Bbb F_{q}}\psi(x)\chi(x).\end{eqnarray*}
Lemma \ref{lem-gausssum} yields the following corollary.
\begin{cor}\label{cor-gausssum}
Let $\psi$ be a multiplicative character and $\chi$ an additive character of $\Bbb F_{q}$. Then $\widehat{G}(\psi,\chi)$ satisfies
\begin{eqnarray*} \widehat{G}(\psi,\chi)=\left\{
\begin{array}{ll}
  q,   &      \mbox{if}\ \psi=\psi_0,\chi=\chi_0,\\
0, & \mbox{if}\ \psi=\psi_0,\chi\neq\chi_0,\\
0, &  \mbox{if}\ \psi\neq\psi_0,\chi=\chi_0.\\
\end{array} \right. \end{eqnarray*}
If $\psi\neq\psi_0,\chi\neq\chi_0$, then
\begin{eqnarray*}|\widehat{G}(\psi,\chi)|=\sqrt{q}.\end{eqnarray*}
\end{cor}

\section{Generalized Jacobi sums and related character sums}
In this section, we present a generalization of Jacobi sums.

Let $k$ be any positive integer. For each integer $1\leq i \leq k$, let $m_i$ be any positive integer, $\lambda_i$ a multiplicative character of $\Bbb F_{q^{m_i}}$, $\chi_i$ the canonical additive character of $\Bbb F_{q^{m_i}}$, and $\Tr_{q^{m_i}/q}$ the trace function from $\Bbb F_{q^{m_i}}$ to $\Bbb F_{q}$. Let $\chi$ denote the canonical additive character of $\Bbb F_{q}$.

Now we define the \emph{generalized Jacobi sums} by
$$
\widehat{J}_{a}(\lambda_{1},\ldots,\lambda_{k})=\sum_{ (c_1, \cdots, c_k) \in \widehat{S}}\lambda_{1}(c_1)\cdots \lambda_k(c_k),
$$
where
$$\widehat{S}=\{(c_1, \cdots, c_k)\in \Bbb F_{q^{m_1}}\times \cdots \times \Bbb F_{q^{m_k}}:\Tr_{q^{m_1}/q}(c_{1})+\cdots+\Tr_{q^{m_k}/q}(c_{k})=a\}$$
for any
$a\in \Bbb F_{q}^{*}$. Note that $|\widehat{S}|=q^{m_1+\cdots+m_k-1}$. If $m_1=m_2=\cdots=m_k=1$ and $a=1$, then $\widehat{J}_{a}(\lambda_{1},\ldots,\lambda_{k})$ is the usual Jacobi sum.

\begin{thm}\label{th-main1}
Let $\lambda_i$ be a multiplicative character of $\Bbb F_{q^{m_i}}$ for $i=1,2,\ldots,k$. Assume that $\widehat{\Bbb F}_{q^{m_i}}^{*}=\langle\phi_i\rangle$ and $\lambda_i=\phi_{i}^{t_i}$, where $i=1,2,\ldots,k$ and $0\leq t_i \leq q^{m_i}-2$.
\begin{enumerate}
\item[(1)] If all the multiplicative characters $\lambda_1,\ldots,\lambda_k$ are trivial, then
$$\widehat{J}_{a}(\lambda_{1},\ldots,\lambda_{k})=q^{m_1+\cdots+m_k-1}.$$
\item[(2)] If some, but not all, of $\lambda_1,\ldots,\lambda_k$ are trivial, then
$$\widehat{J}_{a}(\lambda_{1},\ldots,\lambda_{k})=0.$$
\item[(3)] If all the multiplicative characters $\lambda_1,\ldots,\lambda_k$ are nontrivial and $t_1+\cdots+t_k\equiv 0\pmod{q-1}$, then
$$|\widehat{J}_{a}(\lambda_{1},\ldots,\lambda_{k})|=q^{\frac{m_1+\cdots+m_k-2}{2}}.$$
\item[(4)] If all the multiplicative characters $\lambda_1,\ldots,\lambda_k$ are nontrivial and $t_1+\cdots+t_k\not\equiv 0\pmod{q-1}$, then
$$|\widehat{J}_{a}(\lambda_{1},\ldots,\lambda_{k})|=q^{\frac{m_1+\cdots+m_k-1}{2}}.$$
\end{enumerate}
\end{thm}

\begin{proof}
By the orthogonal relation of additive characters, we have
\begin{eqnarray}\label{eqn-1}
& &\nonumber \widehat{J}_{a}(\lambda_{1},\ldots,\lambda_{k})\\
\nonumber&=&\frac{1}{q}\sum_{\substack{(c_1,\ldots,c_k)\in\Bbb F_{q^{m_1}}\\\times \cdots \times \Bbb F_{q^{m_k}}} }\lambda_{1}(c_1)\cdots \lambda_k(c_k)\sum_{y\in \Bbb F_{q}}\chi\left(y\left(\Tr_{q^{m_1}/q}(c_{1})+\cdots+\Tr_{q^{m_k}/q}(c_{k})-a\right)\right)\\
\nonumber &=&\frac{1}{q}\left(\sum_{c_1\in\Bbb F_{q^{m_1}}}\lambda_{1}(c_1)\right)\cdots\left(\sum_{c_k\in\Bbb F_{q^{m_k}}}\lambda_{k}(c_k)\right)\sum_{y\in \Bbb F_{q}^{*}}\chi\left(y\left(\Tr_{q^{m_1}/q}(c_{1})+\cdots+\Tr_{q^{m_k}/q}(c_{k})-a\right)\right)\\
\nonumber & &+\frac{1}{q}\left(\sum_{c_1\in\Bbb F_{q^{m_1}}}\lambda_{1}(c_1)\right)\cdots\left(\sum_{c_k\in\Bbb F_{q^{m_k}}}\lambda_{k}(c_k)\right)\\
\nonumber &=&\frac{1}{q}\left(\sum_{c_1\in\Bbb F_{q^{m_1}}}\lambda_{1}(c_1)\right)\cdots\left(\sum_{c_k\in\Bbb F_{q^{m_k}}}\lambda_{k}(c_k)\right)\sum_{y\in \Bbb F_{q}^{*}}\chi_1(yc_1)\cdots\chi_k(yc_k)\overline{\chi}(ay)\\
\nonumber & &+\frac{1}{q}\left(\sum_{c_1\in\Bbb F_{q^{m_1}}}\lambda_{1}(c_1)\right)\cdots\left(\sum_{c_k\in\Bbb F_{q^{m_k}}}\lambda_{k}(c_k)\right)\\
\nonumber &=&\frac{1}{q}\sum_{y\in \Bbb F_{q}^{*}}\overline{\chi}(ay)\overline{\lambda_1\cdots\lambda_k}(y)\left(\sum_{c_1\in\Bbb F_{q^{m_1}}}\lambda_1(yc_1)\chi_1(yc_1)\right)\cdots\left(\sum_{c_k\in\Bbb F_{q^{m_k}}}\lambda_k(yc_k)\chi_k(yc_k)\right)\\
\nonumber & &+\frac{1}{q}\left(\sum_{c_1\in\Bbb F_{q^{m_1}}}\lambda_{1}(c_1)\right)\cdots\left(\sum_{c_k\in\Bbb F_{q^{m_k}}}\lambda_{k}(c_k)\right)\\
\nonumber &=&\frac{1}{q}\widehat{G}(\lambda_1,\chi_1)\cdots\widehat{G}(\lambda_k,\chi_k)\sum_{y\in \Bbb F_{q}^{*}}\overline{\chi}(ay)\overline{\lambda_1\cdots\lambda_k}(y)+\frac{1}{q}\left(\sum_{c_1\in\Bbb F_{q^{m_1}}}\lambda_{1}(c_1)\right)\cdots\left(\sum_{c_k\in\Bbb F_{q^{m_k}}}\lambda_{k}(c_k)\right)\\
\nonumber &=&\frac{1}{q}(\lambda_1\cdots\lambda_k)(a)\widehat{G}(\lambda_1,\chi_1)\cdots\widehat{G}(\lambda_k,\chi_k)\sum_{y\in \Bbb F_{q}^{*}}\overline{\chi}(ay)\overline{\lambda_1\cdots\lambda_k}(ay)\\
\nonumber& &+\frac{1}{q}\left(\sum_{c_1\in\Bbb F_{q^{m_1}}}\lambda_{1}(c_1)\right)\cdots\left(\sum_{c_k\in\Bbb F_{q^{m_k}}}\lambda_{k}(c_k)\right).\\
\end{eqnarray}
Assume that $\widehat{\Bbb F}_{q^{m_i}}^{*}=\langle\phi_i\rangle$ and $\lambda_i=\phi_{i}^{t_i}$, where $i=1,2,\ldots,k$ and $0\leq t_i \leq q^{m_i}-2$. Let $\widehat{\Bbb F}_{q}^{*}=\langle\psi\rangle$. For $y\in \Bbb F_{q}^{*}$, one can deduce that
$$\phi_{i}(y)=\psi(y)\text{ for }i=1,2,\ldots,k.$$
Hence, $(\lambda_1\cdots\lambda_k)(y)=\psi^{t_1+\cdots+t_k}(y)$ where $y\in \Bbb F_{q}^{*}$ and $0\leq t_i \leq q^{m_i}-2$ for $i=1,2,\ldots,k$. This implies that
\begin{eqnarray}\label{eqn-2}
\nonumber \sum_{y\in \Bbb F_{q}^{*}}\overline{\chi}(ay)\overline{\lambda_1\cdots\lambda_k}(ay)&=&\sum_{y\in \Bbb F_{q}^{*}}\overline{\chi}(y)\overline{\lambda_1\cdots\lambda_k}(y)\\
\nonumber &=&\sum_{y\in \Bbb F_{q}^{*}}\overline{\chi}(y)\overline{\psi}^{t_1+\cdots+t_k}(y)\\
&=&G(\overline{\psi}^{t_1+\cdots+t_k},\overline{\chi}).
\end{eqnarray}
Combining Equations (\ref{eqn-1}) and (\ref{eqn-2}), we have
\begin{eqnarray}\label{eqn-3}
\nonumber \widehat{J}_{a}(\lambda_{1},\ldots,\lambda_{k})&=&\frac{1}{q}(\lambda_1\cdots\lambda_k)(a)\widehat{G}(\lambda_1,\chi_1)\cdots\widehat{G}(\lambda_k,\chi_k)
G(\overline{\psi}^{t_1+\cdots+t_k},\overline{\chi})\\
& &+\frac{1}{q}\left(\sum_{c_1\in\Bbb F_{q^{m_1}}}\lambda_{1}(c_1)\right)\cdots\left(\sum_{c_k\in\Bbb F_{q^{m_k}}}\lambda_{k}(c_k)\right)
\end{eqnarray}
where $0\leq t_i \leq q^{m_i}-2$ for $i=1,2,\ldots,k$. In the following, we discuss the absolute values of $\widehat{J}_{a}(\lambda_{1},\ldots,\lambda_{k}),a\in \Bbb F_{q}^{*}$, in several cases.
\begin{enumerate}
\item[(1)] If all the multiplicative characters $\lambda_1,\ldots,\lambda_k$ are trivial, we have
$$\widehat{J}_{a}(\lambda_{1},\ldots,\lambda_{k})=|\widehat{S}|=q^{m_1+\cdots+m_k-1}.$$
\item[(2)] If some, but not all, of $\lambda_1,\ldots,\lambda_k$ are trivial, then by Equations (\ref{eqn-sum}), (\ref{eqn-3}) and Corollary (\ref{cor-gausssum}) we have
$$\widehat{J}_{a}(\lambda_{1},\ldots,\lambda_{k})=0.$$
\item[(3)] Assume that all the multiplicative characters $\lambda_1,\ldots,\lambda_k$ are nontrivial. By Equations (\ref{eqn-sum}) and (\ref{eqn-3}), we have
$$\widehat{J}_{a}(\lambda_{1},\ldots,\lambda_{k})=\frac{1}{q}(\lambda_1\cdots\lambda_k)(a)\widehat{G}(\lambda_1,\chi_1)\cdots\widehat{G}(\lambda_k,\chi_k)
G(\overline{\psi}^{t_1+\cdots+t_k},\overline{\chi}).$$
Now we discuss the absolute values of $\widehat{J}_{a}(\lambda_{1},\ldots,\lambda_{k})$ in the following two cases.
\begin{enumerate}
\item[$\bullet$] If $t_1+\cdots+t_k\equiv 0\pmod{q-1}$, then by Lemma \ref{lem-gausssum} and Corollary \ref{cor-gausssum} we have
$$|\widehat{J}_{a}(\lambda_{1},\ldots,\lambda_{k})|=q^{\frac{m_1+\cdots+m_k-2}{2}}.$$
\item[$\bullet$] If $t_1+\cdots+t_k\not\equiv 0\pmod{q-1}$, then by Lemma \ref{lem-gausssum} and Corollary \ref{cor-gausssum} we have
$$|\widehat{J}_{a}(\lambda_{1},\ldots,\lambda_{k})|=q^{\frac{m_1+\cdots+m_k-1}{2}}.$$
\end{enumerate}
\end{enumerate}
The proof is completed.
\end{proof}

Now we define another character sum related to generalized Jacobi sums by
$$
\widetilde{J}_{a}(\lambda_{1},\ldots,\lambda_{k})=\sum_{ (c_1, \cdots, c_k) \in \widetilde{S}}\lambda_{1}(c_1)\cdots \lambda_k(c_k),
$$
where
$$\widetilde{S}=\{(c_1, \cdots, c_k)\in \Bbb F_{q^{m_1}}^{*}\times \ldots \times \Bbb F_{q^{m_k}}^{*}:\Tr_{q^{m_1}/q}(c_{1})+\cdots+\Tr_{q^{m_k}/q}(c_{k})=a\}$$
and
$a\in \Bbb F_{q}$.

\begin{lem}\label{lem-value}
The values of $|\widetilde{S}|$ for $a\in \Bbb F_{q}^{*}$ and $a=0$ are respectively given as follows.
\begin{enumerate}
\item[(1)] For $a\in \Bbb F_{q}^{*}$, $|\widetilde{S}|=\frac{1}{q}\left((q^{m_1}-1)\cdots(q^{m_k}-1)+(-1)^{k+1}\right)$.
\item[(2)] For $a=0$, $|\widetilde{S}|=\frac{1}{q}\left((q^{m_1}-1)\cdots(q^{m_k}-1)+(-1)^{k}(q-1)\right)$.
\end{enumerate}
\end{lem}

\begin{proof}
By the orthogonal relation of additive characters, we have
\begin{eqnarray*}
|\widetilde{S}|&=&\frac{1}{q}\sum_{c_1\in \Bbb F_{q^{m_1}}^{*}}\cdots\sum_{c_k\in \Bbb F_{q^{m_k}}^{*}}\sum_{y\in \Bbb F_{q}}\chi(y(\Tr_{q^{m_1}/q}(c_{1})+\cdots+\Tr_{q^{m_k}/q}(c_{k})-a))\\
&=&\frac{(q^{m_1}-1)\cdots(q^{m_k}-1)}{q}+\frac{1}{q}\sum_{y\in \Bbb F_{q}^{*}}\chi(-ay)\left(\sum_{c_1\in \Bbb F_{q^{m_1}}^{*}}\chi_1(yc_1)\right)\cdots\left(\sum_{c_k\in \Bbb F_{q^{m_k}}^{*}}\chi_k(yc_k)\right)\\
&=&\left\{
\begin{array}{ll}
  \frac{1}{q}\left((q^{m_1}-1)\cdots(q^{m_k}-1)+(-1)^{k+1}\right),   &      \mbox{if}\ a\in \Bbb F_{q}^{*},\\
\frac{1}{q}\left((q^{m_1}-1)\cdots(q^{m_k}-1)+(-1)^{k}(q-1)\right), & \mbox{if}\ a=0.\\
\end{array} \right.
\end{eqnarray*} This proves the conclusions.
\end{proof}

The relationship between $\widetilde{J}_{a}(\lambda_{1},\ldots,\lambda_{k})$ and $\widehat{J}_{a}(\lambda_{1},\ldots,\lambda_{k})$, $a\in \Bbb F_{q}^{*}$, is established as follows.

\begin{lem}\label{lem-relation}
The relationship between $\widetilde{J}_{a}(\lambda_{1},\ldots,\lambda_{k})$ and $\widehat{J}_{a}(\lambda_{1},\ldots,\lambda_{k})$, $a\in \Bbb F_{q}^{*}$, is given as follows.
\begin{enumerate}
\item[(1)] If all of $\lambda_{1},\ldots,\lambda_{k}$ are nontrivial, then
$$\widetilde{J}_{a}(\lambda_{1},\ldots,\lambda_{k})=\widehat{J}_{a}(\lambda_{1},\ldots,\lambda_{k}).$$
\item[(2)] If $\lambda_1,\ldots,\lambda_h$ are nontrivial and $\lambda_{h+1},\ldots,\lambda_k$ are trivial, $1\leq h \leq k-1$, then
$$\widetilde{J}_{a}(\lambda_{1},\ldots,\lambda_{k})=(-1)^{k-h}\widehat{J}_{a}(\lambda_{1},\ldots,\lambda_{h}).$$
\end{enumerate}
\end{lem}

\begin{proof}
The first conclusion is obvious and we only prove the second conclusion. If $\lambda_1,\ldots,\lambda_h$ are nontrivial and $\lambda_{h+1},\ldots,\lambda_k$ are trivial, $1\leq h \leq k-1$, then
\begin{eqnarray*}
\widetilde{J}_{a}(\lambda_{1},\ldots,\lambda_{k})&=&\sum_{ (c_1, \cdots, c_k) \in \widetilde{S}}\lambda_{1}(c_1)\cdots \lambda_k(c_k)\\
&=&\sum_{ (c_1, \cdots, c_k) \in \widetilde{S}}\lambda_{1}(c_1)\cdots \lambda_h(c_h),
\end{eqnarray*}
where $\widetilde{S}=\{(c_1, \cdots, c_k)\in \Bbb F_{q^{m_1}}^{*}\times \ldots \times \Bbb F_{q^{m_k}}^{*}:\Tr_{q^{m_1}/q}(c_{1})+\cdots+\Tr_{q^{m_k}/q}(c_{k})=a\}$. For a fixed $h$-tuple $$(c_1,\ldots,c_h)\in \Bbb F_{q^{m_1}}^{*}\times \cdots \times \Bbb F_{q^{m_h}}^{*},$$
we deduce that the number of the solutions of the equation
$$\Tr_{q^{m_{h+1}}/q}(c_{h+1})+\cdots+\Tr_{q^{m_k}/q}(c_{k})=a-(\Tr_{q^{m_1}/q}(c_{1})+\cdots+\Tr_{q^{m_h}/q}(c_{h}))$$
equals
$$\left\{
\begin{array}{ll}
  \frac{1}{q}\left((q^{m_{h+1}}-1)\cdots(q^{m_k}-1)+(-1)^{k-h+1}\right),   &      \mbox{if}\ \Tr_{q^{m_1}/q}(c_{1})+\cdots+\Tr_{q^{m_h}/q}(c_{h})\neq a,\\
\frac{1}{q}\left((q^{m_{h+1}}-1)\cdots(q^{m_k}-1)+(-1)^{k-h}(q-1)\right), & \mbox{if}\ \Tr_{q^{m_1}/q}(c_{1})+\cdots+\Tr_{q^{m_h}/q}(c_{h})= a,\\
\end{array} \right.$$
by Lemma \ref{lem-value}. Hence,
\begin{eqnarray*}
& &\widetilde{J}_{a}(\lambda_{1},\ldots,\lambda_{k})\\
&=&\frac{1}{q}\left((q^{m_{h+1}}-1)\cdots(q^{m_k}-1)+(-1)^{k-h+1}\right)\sum_{\substack{ (c_1, \cdots, c_h) \in \Bbb F_{q^{m_1}}^{*}\times\cdots\times\Bbb F_{q^{m_h}}^{*}\\ \Tr_{q^{m_1}/q}(c_{1})+\cdots+\Tr_{q^{m_h}/q}(c_{h})\neq a}}\lambda_{1}(c_1)\cdots \lambda_h(c_h)\\
& &+\frac{1}{q}\left((q^{m_{h+1}}-1)\cdots(q^{m_k}-1)+(-1)^{k-h}(q-1)\right)\sum_{\substack{ (c_1, \cdots, c_h) \in \Bbb F_{q^{m_1}}^{*}\times\cdots\times\Bbb F_{q^{m_h}}^{*}\\ \Tr_{q^{m_1}/q}(c_{1})+\cdots+\Tr_{q^{m_h}/q}(c_{h})= a}}\lambda_{1}(c_1)\cdots \lambda_h(c_h)\\
&=&\frac{1}{q}\left((q^{m_{h+1}}-1)\cdots(q^{m_k}-1)+(-1)^{k-h+1}\right)\sum_{\substack{ (c_1, \cdots, c_h) \in \Bbb F_{q^{m_1}}^{*}\times\cdots\times\Bbb F_{q^{m_h}}^{*}}}\lambda_{1}(c_1)\cdots \lambda_h(c_h)\\
& &+(-1)^{k-h}\widetilde{J}_{a}(\lambda_{1},\ldots,\lambda_{h}).
\end{eqnarray*}
Since $\lambda_1,\ldots,\lambda_h$ are nontrivial, we have
\begin{eqnarray*}\sum_{\substack{ (c_1, \cdots, c_h) \in \Bbb F_{q^{m_1}}^{*}\times\cdots\times\Bbb F_{q^{m_h}}^{*}}}\lambda_{1}(c_1)\cdots \lambda_h(c_h)
=\left(\sum_{c_1\in \Bbb F_{q^{m_1}}^{*}}\lambda_{1}(c_1)\right)\cdots\left(\sum_{c_h\in \Bbb F_{q^{m_h}}^{*}}\lambda_{h}(c_h)\right)=0.\end{eqnarray*}
Thus $$\widetilde{J}_{a}(\lambda_{1},\ldots,\lambda_{k})=(-1)^{k-h}\widetilde{J}_{a}(\lambda_{1},\ldots,\lambda_{h})
=(-1)^{k-h}\widehat{J}_{a}(\lambda_{1},\ldots,\lambda_{h}).$$
\end{proof}

Combining Theorem \ref{th-main1} and Lemma \ref{lem-relation} directly yields the following theorem.

\begin{thm}\label{th-main2}
Let $a\in \Bbb F_{q}^{*}$. Let $\lambda_i$ be a multiplicative character of $\Bbb F_{q^{m_i}}$ for $i=1,2,\ldots,k$. Assume that $\widehat{\Bbb F}_{q^{m_i}}^{*}=\langle\phi_i\rangle$ and $\lambda_i=\phi_{i}^{t_i}$, where $i=1,2,\ldots,k$ and $0\leq t_i \leq q^{m_i}-2$.
\begin{enumerate}
\item[(1)] If all the multiplicative characters $\lambda_1,\ldots,\lambda_k$ are trivial, then
$$\widetilde{J}_{a}(\lambda_{1},\ldots,\lambda_{k})=\frac{1}{q}\left((q^{m_1}-1)\cdots(q^{m_k}-1)+(-1)^{k+1}\right).$$
\item[(2)] If $\lambda_1,\ldots,\lambda_h$ are nontrivial, $\lambda_{h+1},\ldots,\lambda_k$ are trivial and $t_1+\cdots+t_h\equiv 0\pmod{q-1}$,$1\leq h \leq k-1$,  then
$$|\widetilde{J}_{a}(\lambda_{1},\ldots,\lambda_{k})|=q^{\frac{m_1+\cdots+m_h-2}{2}}.$$
\item[(3)] If $\lambda_1,\ldots,\lambda_h$ are nontrivial, $\lambda_{h+1},\ldots,\lambda_k$ are trivial and $t_1+\cdots+t_h\not\equiv 0\pmod{q-1}$, $1\leq h \leq k-1$, then
$$|\widetilde{J}_{a}(\lambda_{1},\ldots,\lambda_{k})|=q^{\frac{m_1+\cdots+m_h-1}{2}}.$$
\item[(4)] If all the multiplicative characters $\lambda_1,\ldots,\lambda_k$ are nontrivial and $t_1+\cdots+t_k\equiv 0\pmod{q-1}$, then
$$|\widetilde{J}_{a}(\lambda_{1},\ldots,\lambda_{k})|=q^{\frac{m_1+\cdots+m_k-2}{2}}.$$
\item[(5)] If all the multiplicative characters $\lambda_1,\ldots,\lambda_k$ are nontrivial and $t_1+\cdots+t_k\not\equiv 0\pmod{q-1}$, then
$$|\widetilde{J}_{a}(\lambda_{1},\ldots,\lambda_{k})|=q^{\frac{m_1+\cdots+m_k-1}{2}}.$$
\end{enumerate}
\end{thm}

\section{Nearly optimal codebooks based on generalized Jacobi sums and related character sums}
In this section, we present a construction of codebooks with multiplicative characters of finite fields. We follow the notations in Section 3.

Let $\Bbb F_{q^{m_1}},\ldots,\Bbb F_{q^{m_k}}$, be any $k$ finite fields, where $m_1,\ldots,m_k$ are any $k$ positive integers. For an nonempty set
$$S\subseteq \Bbb F_{q^{m_1}}\times \cdots \times\Bbb F_{q^{m_k}},$$
let $K:=|S|$.

Let $\mathcal{E}_{K}$ denote the set formed by the standard basis of the $K$-dimensional Hilbert space:
\begin{eqnarray*}
& (1,0,0,\cdots,0,0),\\
& (0,1,0,\cdots,0,0),\\
& \vdots \\
& (0,0,0,\cdots,0,1).
\end{eqnarray*}

For any $\lambda_i\in \widehat{\Bbb F}_{q^{m_i}}^{*}$, $i=1,2,\ldots,k$, we define a unit-norm codeword of length $K$ by
\begin{eqnarray}\label{eqn-defn-codeword}\textbf{c}(\lambda_1,\ldots,\lambda_k)=\frac{1}{\sqrt{n_{\bar{\textbf{c}}(\lambda_1,\ldots,\lambda_k)}}}
\left(\lambda_1(c_1)\cdots\lambda_k(c_k)\right)_{(c_1,\ldots,c_k)\in S},\end{eqnarray}
where $n_{\bar{\textbf{c}}(\lambda_1,\ldots,\lambda_k)}$ denotes the Euclidean norm of the vector
$$\bar{\textbf{c}}(\lambda_1,\ldots,\lambda_k):=\left(\lambda_1(c_1)\cdots\lambda_k(c_k)\right)_{(c_1,\ldots,c_k)\in S}.$$

Now we present a generic construction of codebooks as
\begin{eqnarray}\label{eqn-mycodebook}
\mathcal{C}=\{\textbf{c}(\lambda_1,\ldots,\lambda_k):\lambda_i\in \widehat{\Bbb F}_{q^{m_i}}^{*}\text{ for }i=1,2,\ldots,k\}\cup\mathcal{E}_{K}.
\end{eqnarray}
We call $S$ the defining set of $\mathcal{C}$. It is clear that $\mathcal{C}$ has $N=\prod_{i=1}^{k}(q^{m_i}-1)+K$ codewords. If the defining set $S$ is properly selected, then $\mathcal{C}$ may have good parameters with respect to the Welch or the Levenshtein bound.
\subsection{When $S=\widehat{S}$}
In the following, we investigate $I_{\max}(\mathcal{C})$ if we select $S=\widehat{S}$, where $\widehat{S}$ is defined in Section 3 for $a\in \Bbb F_{q}^{*}$. Then $K=|\widehat{S}|=q^{m_1+\cdots+m_k-1}$. Now we consider the value of $n_{\bar{\textbf{c}}(\lambda_1,\ldots,\lambda_k)}$ defined in Equation (\ref{eqn-defn-codeword}). It is easy to verify that
\begin{eqnarray}\label{ineqn-bound}|\widetilde{S}|\leq n_{\bar{\textbf{c}}(\lambda_1,\ldots,\lambda_k)} \leq |\widehat{S}|,\end{eqnarray}
where $|\widetilde{S}|=\frac{1}{q}\left((q^{m_1}-1)\cdots(q^{m_k}-1)+(-1)^{k+1}\right)$ by Lemma \ref{lem-value} and $|\widehat{S}|=q^{m_1+\cdots+m_k-1}$. We remark that $n_{\bar{\textbf{c}}(\lambda_1,\ldots,\lambda_k)}$ achieves the lower bound of Inequality (\ref{ineqn-bound}) if all of $\lambda_1,\cdots,\lambda_k$ are nontrivial, and achieves the upper bound of Inequality (\ref{ineqn-bound}) if all of $\lambda_1,\cdots,\lambda_k$ are trivial.

\begin{thm}\label{th-codebook1}
If $S=\widehat{S}$ and $q\geq 4$, the codebook $\mathcal{C}$ in Equation (\ref{eqn-mycodebook}) has parameters
$$\left(\prod_{i=1}^{k}(q^{m_i}-1)+q^{m_1+\cdots+m_k-1},q^{m_1+\cdots+m_k-1}\right)$$ and
$$I_{\max}(\mathcal{C})=\frac{q^{\frac{m_1+\cdots+m_k+1}{2}}}{\prod_{i=1}^{k}(q^{m_i}-1)+(-1)^{k+1}}.$$
\end{thm}

\begin{proof}
Let $\textbf{c}',\textbf{c}''\in \mathcal{C}$ be any two different codewords. Denote $\mathcal{F}=\mathcal{C}\backslash \mathcal{E}_{K}$. Now we calculate the correlation of $\textbf{c}'$ and $\textbf{c}''$ in the following cases.
\begin{enumerate}
\item[(1)] If $\textbf{c}',\textbf{c}''\in \mathcal{E}_{K}$, we directly have $|\textbf{c}'\textbf{c}''^{H}|=0$.
\item[(2)] If $\textbf{c}'\in \mathcal{F},\textbf{c}''\in \mathcal{E}_{K}$ or $\textbf{c}'\in \mathcal{E}_{K},\textbf{c}''\in \mathcal{F}$, we have
$$|\textbf{c}'\textbf{c}''^{H}|=\frac{1}{\sqrt{n_{\bar{\textbf{c}}(\lambda_1,\ldots,\lambda_k)}}}
\left|\lambda_1(c_1)\cdots\lambda_k(c_k)\right|=\frac{1}{\sqrt{n_{\bar{\textbf{c}}(\lambda_1,\ldots,\lambda_k)}}}$$ for some $(c_1,\ldots,c_k)\in \widehat{\mathcal{S}}$ and $\lambda_i\in\widehat{ \Bbb F}_{q^{m_i}}^{*}$, $i=1,2,\ldots,k$. By the Inequality (\ref{ineqn-bound}), we have
$$\frac{1}{\sqrt{q^{m_1+\cdots+m_k-1}}}\leq|\textbf{c}'\textbf{c}''^{H}|\leq \sqrt{\frac{q}{(q^{m_1}-1)\cdots(q^{m_k}-1)+(-1)^{k+1}}}.$$ Both the lower bound and the upper bound of this inequality can be achieved.
\item[(3)] If $\textbf{c}',\textbf{c}''\in \mathcal{F}$, we assume that $\textbf{c}'=\textbf{c}(\lambda_1',\ldots,\lambda_k')$ and $\textbf{c}''=\textbf{c}(\lambda_1'',\ldots,\lambda_k'')$ with $(\lambda_{1}',\ldots,\lambda_{k}')\neq(\lambda_{1}'',\ldots,\lambda_{k}'')$ where $\lambda_i',\lambda_{i}''\in \widehat{\Bbb F}_{q^{m_i}}^{*}$ for $i=1,2,\ldots,k$.  Denote $\lambda_{i}=\lambda_i'\overline{\lambda_i''}$ for all $i=1,2,\ldots,k$. Then by Equation (\ref{eqn-defn-codeword}) we have
    \begin{eqnarray*}\textbf{c}'\textbf{c}''^{H}
    &=&\frac{1}{\sqrt{n_{\bar{\textbf{c}}(\lambda_1',\ldots,\lambda_k')}}}\cdot
    \frac{1}{\sqrt{n_{\bar{\textbf{c}}(\lambda_1'',\ldots,\lambda_k'')}}}\sum_{\substack{(c_1,\ldots,c_k)\in \widehat{S}}}\lambda_{1}(c_1)\cdots \lambda_k(c_k)\\
    &=&\frac{1}{\sqrt{n_{\bar{\textbf{c}}(\lambda_1',\ldots,\lambda_k')}}}\cdot
    \frac{1}{\sqrt{n_{\bar{\textbf{c}}(\lambda_1'',\ldots,\lambda_k'')}}}\widehat{J}_{a}(\lambda_{1},\ldots,\lambda_{k}).\end{eqnarray*}
    Since $(\lambda_{1}',\ldots,\lambda_{k}')\neq(\lambda_{1}'',\ldots,\lambda_{k}'')$, not all of $\lambda_{i},i=1,2,\cdots,k$ are trivial characters. Hence, by Theorem \ref{th-main1}, we have
    $$|\textbf{c}'\textbf{c}''^{H}|\in \left\{0,
    \frac{q^{\frac{m_1+\cdots+m_k-2}{2}}}{\sqrt{n_{\bar{\textbf{c}}(\lambda_1',\ldots,\lambda_k')}}\cdot \sqrt{n_{\bar{\textbf{c}}(\lambda_1'',\ldots,\lambda_k'')}}},
    \frac{q^{\frac{m_1+\cdots+m_k-1}{2}}}{\sqrt{n_{\bar{\textbf{c}}(\lambda_1',\ldots,\lambda_k')}}\cdot \sqrt{n_{\bar{\textbf{c}}(\lambda_1'',\ldots,\lambda_k'')}}}\right\}.$$ Due to the lower bound of Inequality (IV.3), we obtain
    \begin{eqnarray}\label{ineqn-proof}
    |\textbf{c}'\textbf{c}''^{H}|\leq \frac{q^{\frac{m_1+\cdots+m_k+1}{2}}}{(q^{m_1}-1)\cdots(q^{m_k}-1)+(-1)^{k+1}}.
    \end{eqnarray}
    In the following, we prove that there indeed exist $\textbf{c}',\textbf{c}''\in \mathcal{F}$ such that $|\textbf{c}'\textbf{c}''^{H}|$ achieves the upper bound in Inequality (\ref{ineqn-proof}). Due to Theorem \ref{th-main1} and Inequality (\ref{ineqn-bound}), it is sufficient to prove that there exist $(\lambda_{1}',\ldots,\lambda_{k}'),(\lambda_{1}'',\ldots,\lambda_{k}'')$ such that
    \begin{enumerate}
    \item[$\bullet$] all of $\lambda_{1}',\ldots,\lambda_{k}'$ are nontrivial,
    \item[$\bullet$] all of $\lambda_{1}'',\ldots,\lambda_{k}''$ are nontrivial,
    \item[$\bullet$] all of $\lambda_{1},\ldots,\lambda_{k}$ are nontrivial and $t_1+\cdots+t_k\not\equiv 0\pmod{q-1}$, where $\lambda_{i}=\lambda_i'\overline{\lambda_i''}$, $\widehat{\Bbb F}_{q^{m_i}}^{*}=\langle\phi_i\rangle$ and $\lambda_i=\phi_{i}^{t_i}$ for $i=1,2,\ldots,k$ and $0\leq t_i \leq q^{m_i}-2$.
    \end{enumerate}
    In fact, we can choose a positive integer $s$ such that $$0<s\leq \min\{q^{m_i}-2:i=1,\cdots,k\}\text{ and }0<s+1\leq \min\{q^{m_i}-2:i=1,\cdots,k\}.$$
    Firstly, we assume that $k=2t+1$ for some nonnegative integer $t$. Let
    $$(\lambda_1',\cdots,\lambda_t',\lambda_{t+1}',\cdots,\lambda_{2t}',\lambda_{2t+1}')=(\phi_{1}^{s},\cdots,\phi_{t}^{s},\phi_{t+1}^{s+1},
    \cdots,\phi_{2t}^{s+1},\phi_{2t+1}^{s})$$ and $$(\lambda_1'',\cdots,\lambda_t'',\lambda_{t+1}'',\cdots,\lambda_{2t}'',\lambda_{2t+1}'')=(\phi_{1}^{s+1},\cdots,\phi_{t}^{s+1},\phi_{t+1}^{s},
    \cdots,\phi_{2t}^{s},\phi_{2t+1}^{s+1}).$$ Then
    $$(\lambda_1,\cdots,\lambda_t,\lambda_{t+1},\cdots,\lambda_{2t},\lambda_{2t+1})=(\phi_{1}^{-1},\cdots,\phi_{t}^{-1},\phi_{t+1}^{1},
    \cdots,\phi_{2t}^{1},\phi_{2t+1}^{-1}),$$ which implies that
    $$t_1+\cdots+t_k\equiv -1\pmod{q-1}.$$
   Secondly, we assume that $k=2t$ for some positive integer $t$. Let
    $$(\lambda_1',\cdots,\lambda_{t-1}',\lambda_{t}',\cdots,\lambda_{2t-2}',\lambda_{2t-1}',\lambda_{2t}')=(\phi_{1}^{s},\cdots,\phi_{t-1}^{s},\phi_{t}^{s+1},
    \cdots,\phi_{2t-2}^{s+1},\phi_{2t-1}^{s},\phi_{2t}^{s})$$ and $$(\lambda_1'',\cdots,\lambda_{t-1}'',\lambda_{t}'',\cdots,\lambda_{2t-2}'',\lambda_{2t-1}'',\lambda_{2t}'')
    =(\phi_{1}^{s+1},\cdots,\phi_{t-2}^{s+1},\phi_{t-1}^{s},
    \cdots,\phi_{2t-2}^{s},\phi_{2t-1}^{s+1},\phi_{2t}^{s+1}).$$ Then
    $$(\lambda_1,\cdots,\lambda_{t-2},\lambda_{t-1},\cdots,\lambda_{2t-2},\lambda_{2t-1},\lambda_{2t})=(\phi_{1}^{-1},\cdots,\phi_{t-2}^{-1},\phi_{t-1}^{1},
    \cdots,\phi_{2t-2}^{1},\phi_{2t-1}^{-1},\phi_{2t}^{-1}),$$ which implies that
    $$t_1+\cdots+t_k\equiv -2\pmod{q-1}.$$ Thus we have proved that there indeed exist $\textbf{c}',\textbf{c}''\in \mathcal{F}$ such that
    $$|\textbf{c}'\textbf{c}''^{H}|=\frac{q^{\frac{m_1+\cdots+m_k+1}{2}}}{(q^{m_1}-1)\cdots(q^{m_k}-1)+(-1)^{k+1}}.$$
\end{enumerate}
Summarizing the conclusions in the three cases above, we obtain
$$I_{\max}(\mathcal{C})=\frac{q^{\frac{m_1+\cdots+m_k+1}{2}}}{\prod_{i=1}^{k}(q^{m_i}-1)+(-1)^{k+1}}$$ for $q\geq 4$.
\end{proof}
Theorem \ref{th-codebook1} contains Theorem 19 in \cite{HDY} as a special case. If not all of $m_i,i=1,\cdots,k$, are equal to 1, then the parameters of $\mathcal{C}$ in Theorem \ref{th-codebook1} are different to those of the codebook in Theorem 19 of \cite{HDY}.

\begin{cor}\label{cor}
For the codebook $\mathcal{C}$ in Theorem \ref{th-codebook1}, the followings hold.
\begin{enumerate}
\item If $k=1$ and $m_1=2$, then $\mathcal{C}$ is nearly optimal with respect to the Levenshtein bound.
\item If $k=1,m_1>2$ or $k>1$, then $\mathcal{C}$ is nearly optimal with respect to the Welch bound.
\end{enumerate}
\end{cor}
\begin{proof}
If $k=1$ and $m_1=2$, then $\mathcal{C}$ is a $(q^{2}+q-1,q)$ codebook with $I_{\max}(\mathcal{C})=\frac{1}{\sqrt{q}}$. In this case, we have $N>K^{2}$ and $I_{L}=\frac{\sqrt{q+2}}{q+1}$ by Lemma 1.2. Then we have $\lim_{q\rightarrow \infty}\frac{I_{\max}(\mathcal{C})}{I_{L}}=1$ which implies that $\mathcal{C}$ is nearly achieving the Levenshtein bound.

If $k=1,m_1>2$ or $k>1$, we can deduce that $N<K^{2}$ and
$$I_{W}=\sqrt{\frac{\prod_{i=1}^{k}(q^{m_i}-1)}{(\prod_{i=1}^{k}(q^{m_i}-1)+q^{m_1+\cdots+m_k-1}-1)q^{m_1+\cdots+m_k-1}}}$$
by Lemma 1.1. By Theorem \ref{th-codebook1}, it is easy  to see that $\lim_{q\rightarrow \infty}\frac{I_{\max}(\mathcal{C})}{I_{W}}=1$ which implies that $\mathcal{C}$ is nearly achieving the Welch bound.
\end{proof}

\begin{rem}
In Theorem \ref{th-codebook1}, let $k=2$ and $m_1=1,m_2=2$. Then $\mathcal{C}$ has parameters $(q^3-q+1,q^2)$ and $$I_{\max}(\mathcal{C})=\frac{q^2}{(q-1)(q^2-1)-1}.$$ For a $(q^3-q+1,q^2)$ codebook, Lemma \ref{lem-I.1} implies that $$I_W=\sqrt{\frac{q-1}{q^3}}.$$
In Table I, we list the parameters of some new specific codebooks for $k=2$ and $m_1=1,m_2=2$. From this table, we know that $I_{\max}$ becomes very small for large enough $q$. It can been seen that $I_{\max}$ is very close to $I_W$ given by the Welch bound for large enough $q$, which means that our codebooks are indeed nearly optimal. In particular, the larger the value of $q$ is, the smaller the difference between $I_W/I_{\max}$ and 1 is. These demonstrate that our codebooks should have a good applicability in communications.

\[ \begin{tabular} {c} Table I. Some new nearly optimal codebooks for $k=2$\\ and $m_1=1,m_2=2$ in Theorem \ref{th-codebook1}\\
{\begin{tabular}{cccccc}
  \hline\hline
 $q$ & $N$ & $K$ &  $I_{\max}$ & $I_W$ & $I_W/I_{\max}$\\
  \hline
  $4$ & $61$ & $16$ &  $0.363636$ & $0.216506$ & $0.595392$\\
  $5$ & $121$ & $25$ & $0.263158$ & $0.178885$ & $0.679763$\\
  $7$ & $337$ & $49$ & $0.170732$ & $0.132260$ & $0.774664$\\
  $9$ & $721$ & $81$ & $0.126761$ & $0.104756$ & $0.826406$\\
  $11$ & $1321$ & $121$ & $0.100917$ & $0.086678$ & $0.858904$\\
  $13$ & $2185$ & $169$ & $0.083871$ & $0.073905$ & $0.881175$\\
  $16$ & $4081$ & $256$ & $0.066946$ & $0.060515$ & $0.903938$\\
  $23$ & $12145$ & $529$ & $0.045545$ & $0.042523$ & $0.933648$\\
  $49$ & $117601$ & $2401$ & $0.020842$ & $0.020199$ & $0.969149$\\
  $81$ & $531361$ & $6561$ & $0.012502$ & $0.012269$ & $0.981381$\\
  $121$ & $1771441$ & $14641$ & $0.008334$ & $0.008230$ & $0.987521$\\
  \hline\hline
\end{tabular}}
\end{tabular}
\]
\end{rem}

\subsection{When $S=\widetilde{S}$}
In the following, we investigate $I_{\max}(\mathcal{C})$ if we select $S=\widetilde{S}$, where $\widetilde{S}$ is defined in Section 3 for $a\in \Bbb F_{q}^{*}$. Then $K=|\widetilde{S}|=\frac{1}{q}\left((q^{m_1}-1)\cdots(q^{m_k}-1)+(-1)^{k+1}\right)$ by Lemma \ref{lem-value}.  It is clear that $n_{\bar{\textbf{c}}(\lambda_1,\ldots,\lambda_k)}=K$ in Equation (\ref{eqn-defn-codeword}) for any $\lambda_i\in \widehat{\Bbb F}_{q^{m_i}}^{*}$, $i=1,2,\cdots,k$.

\begin{thm}\label{th-mycodbook2}
If $S=\widetilde{S}$, the codebook $\mathcal{C}$ in Equation (\ref{eqn-mycodebook}) has parameters
$$\left(\prod_{i=1}^{k}(q^{m_i}-1)+\frac{1}{q}\left(\prod_{i=1}^{k}(q^{m_i}-1)+(-1)^{k+1}\right),
\frac{1}{q}\left(\prod_{i=1}^{k}(q^{m_i}-1)+(-1)^{k+1}\right)\right)$$ and
$$I_{\max}(\mathcal{C})=\frac{q^{\frac{m_1+\cdots+m_k+1}{2}}}{\prod_{i=1}^{k}(q^{m_i}-1)+(-1)^{k+1}}.$$
\end{thm}

\begin{proof}
Let $\textbf{c}',\textbf{c}''\in \mathcal{C}$ be any two different codewords. Denote $\mathcal{F}=\mathcal{C}\backslash \mathcal{E}_{K}$. Now we calculate the correlation of $\textbf{c}'$ and $\textbf{c}''$ in the following cases.
\begin{enumerate}
\item[(1)] If $\textbf{c}',\textbf{c}''\in \mathcal{E}_{K}$, we directly have $|\textbf{c}'\textbf{c}''^{H}|=0$.
\item[(2)] If $\textbf{c}'\in \mathcal{F},\textbf{c}''\in \mathcal{E}_{K}$ or $\textbf{c}'\in \mathcal{E}_{K},\textbf{c}''\in \mathcal{F}$, we have
$$|\textbf{c}'\textbf{c}''^{H}|=\frac{1}{\sqrt{K}}
\left|\lambda_1(c_1)\cdots\lambda_k(c_k)\right|=\frac{1}{\sqrt{K}}$$ for some $(c_1,\ldots,c_k)\in \widetilde{\mathcal{S}}$ and $\lambda_i\in\widehat{ \Bbb F}_{q^{m_i}}^{*}$, $i=1,2,\ldots,k$.
\item[(3)] If $\textbf{c}',\textbf{c}''\in \mathcal{F}$, we assume that $\textbf{c}'=\textbf{c}(\lambda_1',\ldots,\lambda_k')$ and $\textbf{c}''=\textbf{c}(\lambda_1'',\ldots,\lambda_k'')$ with $(\lambda_{1}',\ldots,\lambda_{k}')\neq(\lambda_{1}'',\ldots,\lambda_{k}'')$ where $\lambda_i',\lambda_{i}''\in \widehat{\Bbb F}_{q^{m_i}}^{*}$ for $i=1,2,\ldots,k$.  Denote $\lambda_{i}=\lambda_i'\overline{\lambda_i''}$ for all $i=1,2,\ldots,k$. Then by Equation (\ref{eqn-mycodebook}) we have
    \begin{eqnarray*}\textbf{c}'\textbf{c}''^{H}
    &=&\frac{1}{K}\sum_{\substack{(c_1,\ldots,c_k)\in \widetilde{S}}}\lambda_{1}(c_1)\cdots \lambda_k(c_k)\\
    &=&\frac{1}{K}\widetilde{J}_{a}(\lambda_{1},\ldots,\lambda_{k}).\end{eqnarray*}
    Since $(\lambda_{1}',\ldots,\lambda_{k}')\neq(\lambda_{1}'',\ldots,\lambda_{k}'')$, not all of $\lambda_{i},i=1,2,\cdots,k$ are trivial characters. Hence, by Theorem \ref{th-main2}, we have
    $$|\textbf{c}'\textbf{c}''^{H}|\in \left\{\frac{q^{\frac{m_1+\cdots+m_h-2}{2}}}{K}:1\leq h \leq k\right\}\bigcup \left\{\frac{q^{\frac{m_1+\cdots+m_h-1}{2}}}{K}:1\leq h \leq k\right\}.$$
\end{enumerate}
Summarizing the conclusions in the three cases above, we obtain
$$I_{\max}(\mathcal{C})=\frac{q^{\frac{m_1+\cdots+m_k+1}{2}}}{\prod_{i=1}^{k}(q^{m_i}-1)+(-1)^{k+1}}.$$
\end{proof}
Theorem \ref{th-mycodbook2} contains Theorem 15 in \cite{HDY} as a special case. If not all of $m_i,i=1,\cdots,k$, are equal to 1, then the parameters of $\mathcal{C}$ in Theorem \ref{th-mycodbook2} are different to those of the codebook in Theorem 15 of \cite{HDY}.

\begin{cor}
For the codebook $\mathcal{C}$ in Theorem \ref{th-mycodbook2}, the followings hold.
\begin{enumerate}
\item If $k=1$ and $m_1=2$, or $k=2$ and $m_1=m_2=1$, then $\mathcal{C}$ is nearly optimal according to the Levenshtein bound.
\item In other cases, then $\mathcal{C}$ is nearly optimal according to the Welch bound.
\end{enumerate}
\end{cor}
\begin{proof}
The proof is similar to that in Corollary \ref{cor} and we omit the details here.
\end{proof}

\begin{rem}
In Theorem \ref{th-mycodbook2}, let $k=2$ and $m_1=1,m_2=2$. Then $\mathcal{C}$ has parameters $(q^3-2q,q^2-q-1)$ and $$I_{\max}(\mathcal{C})=\frac{q^2}{(q-1)(q^2-1)-1}.$$ For a $(q^3-2q,q^2-q-1)$ codebook, Lemma \ref{lem-I.1} implies that $$I_W=\frac{q-1}{q^2-q-1}.$$
In Table II, we list the parameters of some new specific codebooks for $k=2$ and $m_1=1,m_2=2$. From this table, we know that $I_{\max}$ becomes very small for large enough $q$. It can been seen that $I_{\max}$ is very close to $I_W$ given by the Welch bound for large enough $q$, which means that our codebooks are indeed nearly optimal. In particular, the larger the value of $q$ is, the smaller the difference between $I_W/I_{\max}$ and 1 is. These demonstrate that our codebooks should have a good applicability in communications.

\[ \begin{tabular} {c} Table II. Some new nearly optimal codebooks for $k=2$\\ and $m_1=1,m_2=2$ in Theorem \ref{th-mycodbook2}\\
{\begin{tabular}{cccccc}
  \hline\hline
 $q$ & $N$ & $K$ &  $I_{\max}$ & $I_W$ & $I_W/I_{\max}$\\
  \hline
  $4$ & $56$ & $11$ &  $0.363636$ & $0.272727$ & $0.750000$\\
  $5$ & $115$ & $19$ & $0.263158$ & $0.210526$ & $0.799998$\\
  $7$ & $329$ & $41$ & $0.170732$ & $0.146341$ & $0.857139$\\
  $9$ & $711$ & $71$ & $0.126761$ & $0.112676$ & $0.888885$\\
  $11$ & $1309$ & $109$ & $0.100917$ & $0.091743$ & $0.909094$\\
  $13$ & $2171$ & $155$ & $0.083871$ & $0.077419$ & $0.923072$\\
  $16$ & $4064$ & $239$ & $0.066946$ & $0.062762$ & $0.937502$\\
  $23$ & $12121$ & $505$ & $0.045545$ & $0.043564$ & $0.956505$\\
  $49$ & $117551$ & $2351$ & $0.020842$ & $0.020417$ & $0.979608$\\
  $81$ & $531279$ & $6479$ & $0.012502$ & $0.012348$ & $0.987682$\\
  $121$ & $1771319$ & $14519$ & $0.008334$ & $0.008265$ & $0.991721$\\
  \hline\hline
\end{tabular}}
\end{tabular}
\]
\end{rem}

\section{Conclusions and remarks}
This paper gave two classes of nearly optimal codebooks based on generalized Jacobi sums and related character sums. The main contributions are the following:
\begin{enumerate}
\item[$\bullet$] We generalized the classical Jacobi sums over finite fields and defined the so-called generalized Jacobi sums. The absolute values of the generalized Jacobi sums were investigated. Besides, some related characters sums derived from generalized Jacobi sums were also studied.
\item[$\bullet$] We obtained a class of nearly optimal codebooks in Theorem \ref{th-codebook1} based on the generalized Jacobi sums. This result contains that in \cite[Theorem 19]{HDY} as a special case.
\item[$\bullet$] We obtained a class of nearly optimal codebooks in Theorem \ref{th-mycodbook2} based on the character sums related to the generalized Jacobi sums. This result contains that in \cite[Theorem 15]{HDY} as a special case.
\end{enumerate}

As pointed out in \cite{LHS}, constructing optimal codebooks with minimal
$I_{\max}$ is very difficult in general. This problem is equivalent
to line packing in Grassmannian spaces \cite{CHS}. In frame theory, such a codebook
with $I_{\max}$ minimized is referred to as a Grassmannian frame \cite{SH}.
The codebooks presented in this paper should have
applications in these areas.
With the framework developed in \cite{LG}, our codebooks can be used to obtain deterministic sensing matrices with small coherence for compressed sensing.

It is natural to consider to generalize the generalized Jacobi sums in this paper to some special rings such as Galois rings. The reader is invited to make further progress in this direction.

\section*{Acknowledgements}
 The author is very grateful to the two reviewers for their valuable comments and suggestions that much improved the quality of this paper.

%\ifCLASSOPTIONcaptionsoff
%  \newpage
%\fi

\end{document}